\documentclass{article}

\usepackage{graphicx}
\usepackage{amssymb}

\newcommand{\keyw}[1]{{\bf #1}}

\newtheorem{theorem}{Theorem}
\newtheorem{lemma}{Lemma}

\newenvironment{proof}{\noindent {\em Proof}.\ }{\proofbox\par\smallskip\par}
\newcommand{\halmos}{\rule{1ex}{1.4ex}}
\newcommand{\proofbox}{\hspace*{\fill}\mbox{$\halmos$}}

\begin{document}

\title{A Space-Optimal Hidden Surface Removal Algorithm for Iso-Oriented Rectangles}

\author{A. Tsakalidis\\ 
Department of Computer Engineering \& Informatics,\\ 
University of Patras, Greece.\\
\texttt{tsak@cti.gr} 
\and K. Tsichlas \\
Informatics Department,\\
Aristotle University of Thessaloniki, Greece.\\ 
\texttt{tsichlas@delab.csd.auth.gr}}

\maketitle

\begin{abstract}
 We investigate the problem of finding the visible pieces of a scene of objects from a specified viewpoint. In particular, we are interested in the design of an efficient hidden surface removal algorithm for a scene comprised of iso-oriented rectangles. We propose an algorithm where given a set of $n$ iso-oriented rectangles we report all visible surfaces in $O((n+k)\log n)$ time and linear space, where $k$ is the number of surfaces reported. The previous best result by Bern, has the same time complexity but uses $O(n\log n)$ space.

{\bf  Computational Geometry, Computer Graphics, Hidden Surface Removal, Iso-Oriented Rectangles.}

\end{abstract}

\section{Introduction} \label{sec:intro}
 
 The {\em Hidden Surface Removal} (HSR) problem is one of the fundamental problems in computer graphics. Given a set of objects in a three dimensional scene we want to compute the visible parts of the scene from a given viewpoint. As a result, pieces of objects that lay behind other objects with respect to the given viewpoint are invisible. In general, points are visible when the line between each point and the viewpoint is not intersected by other objects.
 
 A slightly easier problem than HSR is the Hidden Line Elimination (HLE) problem. In the HLE problem it is assumed that surfaces do not carry information (like color) and only the visible line segments that define the exterior of each object are interesting. These problems are strongly connected but experience has shown that the HSR problem is more difficult the the HLE problem. In particular, hidden line information does not necessarily allow us to determine the frontmost faces of the environment \cite{mck}.
 
 In the general case of the HSR problem the scene may consist of arbitrary objects in the three dimensional space. A simple but important special case of the general HSR problem is one in which the scene consists of $n$ rectangles which are parallel to the $x-y$ plane and their edges are parallel to the $x$ and $y$ axis. These rectangles are called {\em iso-oriented}. An iso-oriented rectangle $R$ can be fully specified by five coordinates, $[R.x_{1}, R.x_{2}]\times[R.y_{1},R.y_{2}]\times R.z$. It is assumed that no pair of rectangles intersects in a two dimensional region, though pairs may intersect along an edge.
 
 In this paper, we consider a static scene (changes are not allowed in the scene) comprised of $n$ iso-oriented rectangles in the three-dimensional space. The goal is to to compute and depict all visible parts of these rectangles as they would be seen from an observer at a specified viewpoint. We assume that this viewpoint lies at $z=+\infty$.

 Most of the algorithms known for scenes of iso-oriented rectangles are {\em output sensitive}. This means that their time complexity depends on the size of the output, that is the complexity of the visible part of the scene. These algorithms are generally more efficient than algorithms with time complexity depending solely on $n$. For static scenes consisting of polygons, McKenna \cite{mck} has designed a worst case optimal algorithm with $O(n^{2})$ time and space complexity. Note that the complexity of the visible scene cannot be larger than $O(n^{2})$ when the polygons in the scene have $n$ edges in total. This is because the number of visible parts can not exceed the number of intersections between the objects.
 
 Gutting et al. \cite{gut} proposed an output sensitive HLE algorithm for static scenes of rectangles with time complexity $O((n+k)\log^{2}n)$, where $k$ is the number of reported visible segments. Note, that when $k$ is small, then the time complexity is far less than $O(n^{2})$. Of course, for complicated scenes where $k=O(n^{2})$ this algorithm is less efficient than the algorithm of McKenna. This algorithm also handles $c$-oriented rectangles (rectangles aligned with a fixed number of orientations, not just horizontal and vertical). Atallah and Goodrich \cite{ata} have proposed an algorithm with $O(n^{3/2}+k)$ time complexity. 
 
 For a static scene consisting of iso-oriented rectangles, Bern \cite{bern} has designed an algorithm with $O((n+k)\log n)$ time complexity and $O(n\log n)$ space complexity. 
He designed algorithms for the HLE problem and later he extended them to tackle the HSR problem. Mehlhron et al. \cite{mehl} proposed an algorithm with $O(n\log n+k\log(n^{2}/k))$ time complexity and $O(n\log n)$ space complexity for the HLE problem. Kitsios et al. \cite{kits} have improved on this result by proposing an algorithm for the HLE problem that uses linear space while retaining the above time complexity. 
 
 This paper extends the last result of Kitsios and Tsakalidis to tackle the HSR problem in a scene consisting of iso-oriented rectangles. We propose an HSR algorithm with $O((n+k)\log{n})$ time complexity using linear space. Our result improves the algorithm proposed by Bern \cite{bern} by a logarithmic factor in its space complexity. Our algorithm modifies the algorithm of \cite{kits} and extends it by adding appropriate data structures to store the necessary surface information. In addition, our algorithm needs only one pass of the scene, while the algorithm of Kitsios et al. needs two, one pass for the vertical edges and one for the horizontal edges (our algorithm can also be used for the HLE problem with minor modifications). The only drawback is that in the HLE problem the multiplicative factor of $k$ in the time complexity is $O(\log{(n^2/k)})$, which in general is less than $O(\log{n})$. This is due to the maintenance of the visible regions and is another indication that the HSR problem is generally more difficult than the HLE problem.
 
 This special case of hidden surface removal has application to overlapping windows in computer displays. It allows us to solve window management problems efficiently. In addition, the algorithms for this restricted case of the hidden surface removal problem could find use in cartographic applications as well as used in VLSI design tools for many-layer technologies.
 
 The remainder of the paper is organized as follows. In Section~\ref{sec:prel} a description of the algorithm is given as well as some basic definitions and techniques, which are essential for the comprehension of the algorithm. The algorithm for the HSR problem is described in Section~\ref{sec:alg}. The description is divided into two parts, the first part describes the preprocessing stage while the second part describes the reporting stage. Finally, in Section~\ref{sec:con} we conclude with some final remarks.

\section{Preliminaries} \label{sec:prel}
 
 The HSR problem is considered in a static scene consisting of $n$ iso-oriented rectangles. Our algorithm uses the plane sweep technique and cuts the scene into slabs in order to guarantee linear space.
 
 Initially, the scene is divided into {\em slabs}. A plane parallel to the $y-z$ plane sweeps each slab along the $x$ axis from $x=-\infty$ to $x=+\infty$. All edges of rectangles parallel to the $x$ axis are called {\em horizontal} while all parallel to the $y$ axis are called {\em vertical}. In this way, the intersection of the sweep plane and the scene in a random position is a set of vertical segments. Assuming, without loss of generality, that all the $x$, $y$ and $z$ coordinates are distinct then the intersection of the sweep plane and the scene in each sweep station consists of one and only one vertical edge. The sweep stations of the algorithm consist of the ordered set of the $x$ coordinates of the rectangles.
 
 The set of the vertical edges of each slab is stored in a segment tree. A segment tree~\cite{bent} is constructed from scratch at the beginning of each slab. This tree has $2n-1$ leaves and it is implemented as a binary balanced tree. Its $i$-th leaf represents the elementary interval $[y_{i},y_{i+1}]$ (if we sort the $y$ coordinates then $y_{i}$ will be just before $y_{i+1}$), which is termed a $y$-range. Each internal node $u$ has a $y$-range equal to the union of the $y$-ranges of the leaves of the subtree rooted at $u$. Henceforth, the ends of a vertical segment $s$ will be represented as $s.y_{1}$ and $s.y_{2}$, where $s.y_{1} < s.y_{2}$. In this way, the $y$-range of $u$ is between $u.y_{1}$ and $u.y_{2}$. Let $father(u)$, $lson(u)$ and $rson(u)$ denote the father and the two children of $u$ respectively. A vertical segment is associated with $O(\log{n})$ nodes $u$ such that $u.range \subseteq [s.y_{1},s.y_{2}]$ and $father(u).yrange \not \subset [s.y_{1},s.y_{2}]$. Each node $u$ of the segment tree has a set $S(u)$ of associated segments. By precomputing the lists of rectangles in each node of the segment tree we are able to reduce the time complexity by a logarithmic factor \cite{bern}.
  
 The horizontal segments are stored in a binary search tree, called the {\em region tree}. The region tree is maintained during the transition between slabs. In this way, information concerning the visible regions is transferred among slabs. The leaves of the region tree are linked by means of a double linked list. We implement the region tree as a red-black tree \cite{red}, so that the amortized cost of updates is constant.
 
 Finally, the coordinates of a rectangle $R=[x_{1},x_{2}]\times[y_{1},y_{2}]\times z$ are denoted by $R_{x_{1}}$, $R_{x_{2}}$, $R_{y_{1}}$, $R_{y_{2}}$ and $R_{z}$ respectively. The line segment defined by the endpoints $(R_{x_{1}},R_{y_{1}},R_{z})$ and $(R_{x_{2}},R_{y_{2}},R_{z})$ is referred as the left edge of the rectangle $R$. The right, bottom and top edges are defined similarly. We say that a rectangle $R$ is higher than a rectangle $R'$, or that $R'$ is lower than $R$, when $R.z > R'.z$. The same goes for the edges. Finally, we assume that there is a fictitious rectangle $background$ that lies behind the whole scene with height $z=-\infty$.

\section{The Algorithm} \label{sec:alg}
 
 The algorithm consists of two stages, the preprocessing stage and the reporting stage. In the former stage the necessary data structures are constructed and initialized appropriately. In the latter stage, we use the available data structures to find all visible surfaces by using the plane sweep technique. In the following, we will first refer to the preprocessing stage and then move to the reporting stage.

\subsection{The Preprocessing Stage} \label{sub:preprocessing}
 
 First of all, the vertices of the rectangles are sorted with respect to their $x$, $y$ and $z$ coordinates. The $y$-order will be used for the construction of the segment tree. The $x$-order will be used for the plane sweep while the $z$-order will be used in depth computations. Then, we cut the scene into slabs. Each slab is defined by two planes normal to the $x$ axis, so that each plane contains $\frac{n}{\log{n}}$ vertical edges. Since there are exactly $2n$ vertical edges, the number of slabs will be $\log{n}$. The set of vertical edges in one slab is called $S_{v}$ while the set of horizontal segments that span the slab is denoted by $S_{h}$. Note that $|S_{v}|=O(\frac{n}{\log{n}})$ by definition, while $|S_{h}|=O(n)$ (an example is depicted in Figure~\ref{fig:example1}).

\begin{figure}
 \centering
 \includegraphics[scale=0.6]{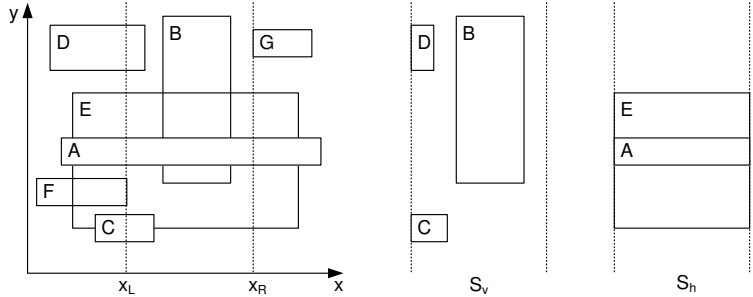}
 \caption{The rectangles as seen from $z=+\infty$ and the clipped rectangles that form sets $S_{v}$ and $S_{h}$.}
 \label{fig:example1}
\end{figure}

 In the beginning of each slab we build an enriched segment tree. In particular, let $T_{v}(u)$ denote the highest segment of $S_{v}$ associated with node $u$ of the segment tree. If there is no such segment then $T_{v}(u)=background$. Every node of the segment tree is augmented with the field $u.H_{h}$, which contains the highest segment of $S_{h}$ that spans $u.yrange$. Note that $u.H_{h}$ remains invariable during the sweep in each slab. Each node $u$ of the segment tree is also augmented with the following fields:
 
\begin{itemize}

 \item{$u.H=u.L=\max_{z}\{u.H_{h}, T_{v}(u)\}$, if $u$ is a leaf}
 \item{$u.H=\max_{z}\{lson(u).H, rson(u).H, T_{v}(u)\}$ and \\
       $u.L=\max_{z}\{\min_{z}\{lson(u).H,rson(u).H\},T_{v}(u)\}$, if $u$ is an internal node.}
\end{itemize}

 Let $S(u)$ be the set consisting of segments in the subtree rooted at $u$ and segments of $S_{h}$ associated with the ancestors of $u$.

\begin{lemma}
 Field $u.L$ is the lowest visible segment in the subscene restricted to $u.yrange$ consisting of segments of $S(u)$. Field $u.H$ is analogously the highest segment among the segments of $S(u)$.
\end{lemma}

\begin{proof}
 Assume a segment $s'$ in the subtree of $u$ such that $s'$ is visible in the subscene restricted to $u.yrange$ and it is lower than $u.L$. Assume that $s'$ is associated with node $w$. Then, the minimum between $lson(father(w)).H$ and $rson(father(w)).H$ will be $s'$ (by assumption). If $T_{v}(father(w))\neq background$ then $s'$ is not visible since it is completely obscured by $T_{v}(father(w))$ which contradicts our assumption. As a result, $father(w).L$ will be $s'$. Applying the same procedure to all ancestors of $w$, we will reach $u$ and $u.L$ will be $s'$, which is a contradiction. As a result, $u.L$ is the lowest visible segment in $S(u)$. In the same way we can prove that $u.H$ is the highest segment among the segments of $S(u)$.
\end{proof}

 The sweep plane traverses the slab from left to right. When a vertical edge $s$ of a rectangle $R\in S_{v}$ is encountered, then the data structures must be appropriately updated and queried to report the visible regions located at the $yrange$ defined by this vertical edge. As a result, fields $u.H$ and $u.L$ are subject to changes. These changes are reflected to the nodes of the segment tree by removing or adding segments to their subtrees. Let $m$ be the number of insertions/deletions in the subtree rooted at a node $u$ and let $s_{1},s_{2},\ldots,s_{m}$ be the sequence of values of field $u.H$ after each such update operation. We precompute this sequence and store it in an array $u.High$ of $m+1$ entries. In a similar manner, we precompute the sequence $u.Low$ for field $u.L$. For both arrays we use a pointer $u.p$ that points to the current value of $u.H$ and $u.L$ in arrays $u.High$ and $u.Low$ respectively. An insertion or deletion of a segment in the subtree rooted at $u$ is simulated by incrementing the pointer $u.p$.
 
 The construction of $u.High$ and $u.Low$ for each node $u$ is based on the auxiliary sequences $u.Top_{v}$, $u.xTop_{v}$ and $u.xHL$. Assume that $x_{1},x_{2},\ldots,x_{m}$ is the ordered sequence of $x$ coordinates of segments of $S_{v}$ that are inserted or deleted to $u$. Assume that when $x=x_{i}$ a segment is inserted or deleted to $S(u)$ possibly changing $T_{v}(u)$ (the highest segment of $S_{v}$ associated to node $u$). The sequence $u.Top_{v}$ has $m+1$ entries that record all changes of $T_{v}(u)$. Specifically, $u.Top_{v}[i]=s$ if $T_{v}(u)=s$ for $x_{i}\leq x < x_{i+1}$ ($x_{0}$ corresponds to the $x$ coordinate of the left vertical plane which defines the current slab, $x_{m+1}$ is defined analogously for the right plane). The sequence $x_{1},x_{2},\ldots,x_{m}$ is stored in the array $u.xTop_{v}$. Similarly, let $x_{1},x_{2},\ldots,x_{p}$ be the sequence of the $x$ coordinates of segments of $S_{v}$ that are inserted into or deleted from nodes in the subtree rooted at $u$. After each insertion or deletion of such a segment the $H$ and $L$ fields of $u$ may change. This sequence is stored in $u.xHL$. The construction of all these sequences is feasible in linear time. This is proved by using the following lemma:
 
\begin{lemma}
 We are given a set of $q$ horizontal line segments in the $t-z$ plane with integer $t$ coordinates in the range $[0,2q-1]$. The segments are given in decreasing $z$ order. In $O(q)$ time it is possible to construct an array $A$ of $2q$ entries that stores in its $i$-th entry the segment with highest $z$ coordinate among segments that span the $t$-interval $[i,i+1]$.
\end{lemma} 

\begin{proof}
 The proof is given in \cite{bern,mehl}.
\end{proof}

\begin{lemma}
 Sequences $u.Top_{v}$ and $u.xTop_{v}$ for each node $u$ of the segment tree can be constructed in linear time. Sequences $u.H_{h}$ for each leaf $u$ of the segment tree can be constructed in linear time.
\end{lemma}

\begin{proof}
 The proof is given in \cite{bern}.
\end{proof}

\begin{lemma} \label{lem:con-seg}
 The sequences $High$, $Low$ and $u.xHL$ for all nodes can be computed in $O(n)$ time.
\end{lemma}

\begin{proof}
 The proof for this lemma is given in \cite{bern}. We must add that during the construction of the segment tree the filling of the $u.H_{h}$ fields for all inner nodes $u$ can be accomplished by setting: $u.H_{h}=\max\{lson(u).H_{h},rson(u).H_{h}\}$
\end{proof}

 Apart from the segment tree we use an auxiliary leaf-oriented balanced binary search tree, which we call the {\em region tree}. The region tree $T$ is used to store the horizontal segments of the scene, ordered according to the $y$ coordinate, as well as the necessary information to report visible regions. A horizontal segment is stored in a leaf of $T$ only when it is intersected by the sweep plane and it is visible. The leaves of this tree form a double linked list. The region tree is dynamic red-black tree \cite{red} and remains the same during the transition between slabs. 
 
 In the region tree, the visible horizontal segments (edges of rectangles) are inserted or deleted during the transition of the sweep plane between the sweep stations. Apart from the edges, the leaves of the tree store the rectangle in which the area between two consecutive visible horizontal segments belongs to. Assume that $f$ is a leaf of $T$, $right(f)$ is the leaf immediately right to $f$ and $f.region$ is the rectangle which owns the region between the horizontal segment stored in $f$ and $right(f)$. If the area between two consecutive segments is not part of a rectangle, then $f.region=background$. The traversal of the double linked list enables us to report visible regions. These regions are defined by the $y$ coordinates of two consecutive horizontal edges and the $x$ coordinates of the start of the region and the current position of the sweep plane. The $x$ coordinate of the start of the region is a field attached to $f.region$ and contains the $x$ coordinate of the sweep station where the field $f.region$ was updated for the last time. In the following, we will analyze the reporting procedure and specify how the region tree is updated.

\subsection{The Reporting Stage}

 At this point we will focus on a single slab. The necessary initialization for each slab is described in \ref{sub:preprocessing}. The description of the reporting stage is split into two parts. In the first part we explore the case where the sweep plane intersects the left edge of a rectangle while in the second part we explore the case where the sweep plane intersects the right edge of a rectangle. First, we are going to explore what happens when a new edge in virtually (due to the preprocessing) inserted in the segment tree - that is the sweep plane intersects a left edge.

\subsubsection{Insertion of a New Edge}

 The procedure $LeftEdge(R,true,root)$ is invoked when the left edge of a rectangle $R$ is encountered. This procedure updates the fields $u.H$ and $u.L$ for all nodes $u$ visited in the segment tree. Assume that the position of the sweep plane is at $x=x_{s}$. In Figures~\ref{fig:left} and \ref{fig:left1} a description of the algorithm for the insertion of the left edge of an arbitrary rectangle $R$ is given.
 
\begin{figure}
{\scriptsize
\begin{tabbing}
\quad \=\quad \=\quad \=\quad \=\quad \=\quad \=\quad \kill
\keyw{Procedure} LeftEdge(Rectangle $R$, boolean $visible$, segment\_tree\_node $u$)  \\
1. \> \> \keyw{if} $(R.z<u.L)$ \keyw{then} $visible=FALSE$ \\
2. \> \> \keyw{if} $(R.y_{1}<u.y_{1})$ AND $(u.y_{2}<R.y_{2})$ \keyw{then} \\
3. \> \> \> \keyw{if} visible \keyw{then} LeftReportRegions($R$,$u$) \\
4. \> \> \keyw{else} \\
5. \> \> \> \keyw{if} $(R.y_{1}<u.y_{mid})$ \keyw{then} LeftEdge($R$,$visible$,$lson(u)$) \\
6. \> \> \> \keyw{if} $(R.y_{2}>u.y_{mid})$ \keyw{then} LeftEdge($R$,$visible$,$rson(u)$) \\
7. \> \> $u.p=u.p+1$
\end{tabbing}
}
 \caption{This procedure is invoked when the left edge of a rectangle $R$ is encountered.}
 \label{fig:left}
\end{figure}

\begin{figure}
{\scriptsize
\begin{tabbing}
\quad \=\quad \=\quad \=\quad \=\quad \=\quad \=\quad \kill
\keyw{Procedure} LeftReportRegions(Rectangle $R$, segment\_tree\_node $u$)  \\
1. \> \> \keyw{if} $(R.z<u.L)$ \keyw{then return} \\
2. \> \> \keyw{if} $(u.H<R.z)$ \keyw{then} \\
3. \> \> \> Find leaves $p$ and $q$ in $T$ so that $p.y<u.y_{1}$ and there is no other leaf $v$ such that $p.y<v.y<u.y_{1}$ \\ 
\> \> \> and act analogously for $u.y_{2}$ \\
4. \> \> \> Output regions formed by the horizontal segments found between $p.y_{1}$ and $q.y_{2}$ \\
5. \> \> \> Remove all horizontal segments found between $p.y_{1}$ and $q.y_{2}$ \\
6. \> \> \> \keyw{if} $(u.y_{1}=R.y_{1})$ \keyw{then} insert in $T$ the segment $s$ with $s.x_{1}=R.x_{1}$ and $s.y=u.y_{1}$ \\
7. \> \> \> \keyw{if} $(u.y_{2}=R.y_{2})$ \keyw{then} insert in $T$ the segment $s$ with $s.x_{1}=R.x_{1}$ and $s.y=u.y_{2}$ \\
8. \> \> \> Update properly the fields region of $p$, $q$ and the newly inserted leaves and then change the $x$ \\ 
\> \> \> field of each such field to be current $x$ \\
9. \> \> \keyw{else} \\
10. \> \> \> LeftReportRegions($R$,$lson(u)$) \\
11. \> \> \> LeftReportRegions($R$,$rson(u)$) 
\end{tabbing}
}
 \caption{This procedure is invoked by procedure $LeftEdge$ to report the visible regions.}
 \label{fig:left1}
\end{figure}

 Assume that the left edge of $R$ is divided into consecutive invisible and visible segments. Let the visible segments be $s(u_{1}),s(u_{2}),\ldots,s(u_{l})$ and the invisible segments $s(w_{1}),s(w_{2}),\ldots,s(w_{m})$, where $u_{i}$ and $w_{i}$ are nodes of the segment tree. The procedure given in Figure~\ref{fig:left} stops the recursive search when one of these nodes is reached. When a node $u_{i}$ is reached, the region tree $T$ is queried with the range $u_{i}.yrange$. The result of this query are two leaves $f_{1}$ and $f_{2}$. Note that it may be the case that $f_{1}$ and $f_{2}$ are the same leaves. By using the double linked list, all leaves between $f_{1}$ and $f_{2}$ are traversed to report the visible regions. After reporting all visible regions that are obscured by the new rectangle, we remove all leaves between $f_{1}$ and $f_{2}$. Finally, we make the necessary adjustments to reflect the fact that this region belongs to the new rectangle. 
 
 In particular, if $f_{1}$ and $f_{2}$ are different leaves then we remove all leaves between them since the new rectangle $R$ will obscure the rectangles they represent. Then, if the $y$-coordinates of the upper and lower edges of $R$ do not belong in $u_{i}.yrange$, we just update the field region between $f_{1}$ and $f_{2}$ so that it belongs to $R$. In any other case we must create a new leaf for either the upper or the lower edge or both updating appropriately the region fields. If $f_{1}$ and $f_{2}$ are the same leaf, then either we have to insert the upper or the lower edge of $R$ or there is a node $u_{i+1}$ such that $u_{i}.yrange$ and $u_{i+1}.yrange$ are adjacent and the same cases apply. The following lemma is essential in the construction of the region tree.
 
\begin{lemma} \label{lem:region}
 A visible region in the region tree is defined by the $x$ coordinate of its insertion, the $x$ coordinate of the sweep plane and the $y$ coordinates of two edges of rectangles.
\end{lemma}

\begin{proof}
 This is trivially true for the $x$ coordinates. We have to show that each visible region is defined between two horizontal edges. This is true since each rectangle is characterized by only one $z$ coordinate. As a result, rectangles will always intersect among their edges.
\end{proof}

\begin{figure}
 \centering
 \includegraphics[scale=0.6]{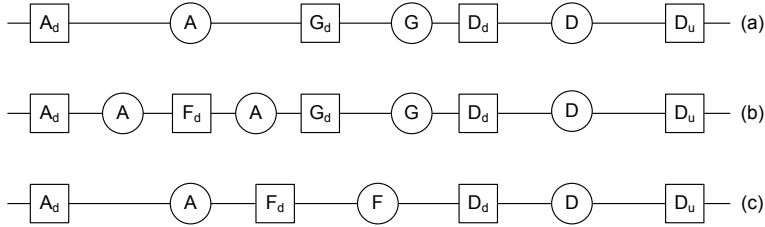}
 \caption{Consecutive parts of the list of leaves of the region tree $T$ before and after the insertion of a horizontal segment $F_{d}$. The circles designate the owner of the region between the edges (field $region$).}
 \label{fig:dlistleft}
\end{figure}

 We give an example of how the region tree is handled. We assume that uppercase letters designate rectangles. Indexes $u$ and $d$ in uppercase letters designate the upper and lower edge of the specified rectangle respectively. In Figure~\ref{fig:dlistleft}, a part of the double linked list of the leaves of the region tree $T$ is depicted. Initially, the sweep plane is at position $x=x_{s-1}$ and the list of the leaves of $T$ is shown in Figure~\ref{fig:dlistleft}(a). The next sweep station of the sweep plane is at $x=x_{s}$. In this position, the sweep plane intersects the left edge of the rectangle $F$. Assume that the range $[F.y_{1},F.y_{2}]$ is between the leaves of $T$  with labels $A_{d}$ and $D_{u}$ as depicted in Figure~\ref{fig:dlistleft}(a). Finally, assume that $A.z < G.z < F.z < D.z$ and that $A.y_{1} < F.y_{1} < G.y_{1}$ and $D.y_{1}<F.y_{2}<D.y_{2}$. Therefore, only the lower edge of $F$ will be inserted in $T$ yielding the list depicted in Figure~\ref{fig:dlistleft}(b). 
 
 The visible regions reported are:
 
\begin{enumerate}
 
 \item{The region of $A$ defined by $F.y_{1}-G.y_{1}$ and $A.x-x_{s}$ ($A.x$ is the $x$ coordinate of the start of the region).}
 
 \item{The region of $G$ defined by $G.y_{1}-D.y_{1}$ and $G.x-x_{s}$.}

\end{enumerate}

 Finally, the regions reported are removed from the region tree $T$ resulting in the list of Figure~\ref{fig:dlistleft}(c). Note that the algorithm given in Figures \ref{fig:left} and \ref{fig:left1} would identify all nodes of the segment tree such that the union of their $y$-ranges would be equal to $[F.y_{1},D.y_{1}]$. Then, the region tree would be updated for each such node.

 The crucial observation in the analysis of procedure $LeftEdge$ is that procedure $LeftReportRegions$ stops its recursive search whenever it reaches one of the $u_{i}$ or $w_{i}$ nodes (Figure~\ref{fig:leftex}(a)). As a result, even a visible segment of the left edge hidden behind a complicated part of the scene costs only $O(\log{n})$ to discover. Procedure $LeftReportRegions$ explores a forest of subtrees of the segment tree. The roots of these subtrees are nodes that list $R$ ($R$ is associated with these nodes), the nodes $u_{1},u_{2},\ldots,u_{l}$ are leaves of these subtrees and the remaining leaves are nodes $w_{1},w_{2},\ldots,w_{m}$ (the proof of this argument can be found in \cite{bern}). For each of the $u_{i}$ nodes we search the tree $T$ in $O(\log{n})$ time locating leaves $p$ and $q$ (the search keys are the $y$ coordinates of the endpoints of the $u_{i}.yrange$). When we find the leaves we may insert at most two new horizontal segments ($O(1)$ amortized time), then report all the regions which are defined between $p$ and $q$ in the double linked list and finally remove all the reported regions while updating the new region. The deletion of these regions (leaves) also incurs an $O(1)$ amortized time cost per leaf. As a result, the cost for each node $u_{i}$ is $O(k_{i}+\log{n})$, where $k_{i}$ is the number of reported regions. 
 
\begin{lemma} \label{lem:left}
 The procedure $LeftEdge$ requires $O(k\log{n})$ time to report $k$ visible regions for a scene of $n$ rectangles taking only into account the left edges.
\end{lemma} 

\begin{proof}
 Assume that the visible segments of the left edge of a rectangle $R_{j}$ are $s(u_{1}),s(u_{2}),\ldots,s(u_{l_{j}})$ and the invisible segments are $s(w_{1}),s(w_{2}),\ldots,s(w_{m})$, where $u_{i}$ and $w_{i}$ are nodes of the segment tree. The discovery of the $u_{i}$ nodes is achieved in $O(\log{n})$ time. For every node $u_{i}$, $O(\log{n})$ time is required to search the region tree. The cost of a single edge will be $O(\sum_{i=1}^{l_{j}}{O(\log{n})}+k_{j})$, where $k_{j}$ is the number of reported regions. This means that the cost for each visible segment in the segment tree is $O(\log{n})$ while the cost for each reported region is amortized $O(1)$. Since the number of reported regions is at least equal to the number of reported segments we assume that the cost for each region is $O(\log{n})$.
 
 Thus, for all rectangles the time complexity will be:
 
 \[ \sum_{j=1}^{n}O(l'_{j}\log{n})=O(k\log n) \]
where $k=\sum_{j=1}^{n}{l'_{j}}$ is the number of visible regions.
\end{proof}

 \begin{figure}
 \centering
 \includegraphics[scale=0.8]{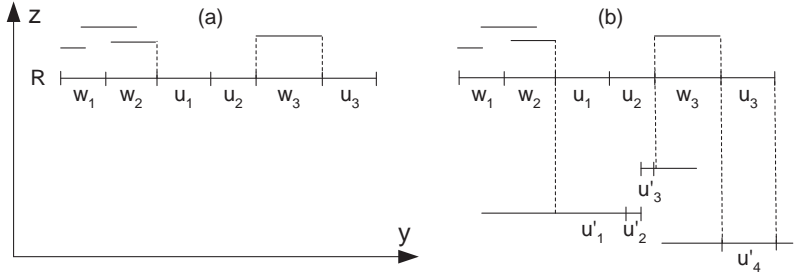}
 \caption{(a) $R$'s left edge divided into visible ($u_{i}$) and invisible ($w_{i}$) pieces, (b) $R$'s right edge and revealed rectangles beneath $R(u'_{i})$.}
 \label{fig:leftex}
\end{figure}

\subsubsection{Deletion of an Edge}

The procedure $RightEdge(R,true,background,root)$ is invoked when the right edge of a rectangle $R$ is encountered. This procedure is depicted in Figures~\ref{fig:right} and \ref{fig:right1}.

\begin{figure}
{\scriptsize
\begin{tabbing}
\quad \=\quad \=\quad \=\quad \=\quad \=\quad \=\quad \kill
\keyw{Procedure} RightEdge(Rectangle $R$, boolean $visible$, Rectangle $R'$, segment\_tree\_node $u$)  \\
1. \> \> \keyw{if} $(R.z<u.L)$ \keyw{then} $visible=FALSE$ \\
2. \> \> \keyw{if} (($R.y_{1}\leq u.y_{1}$) AND ($u.y_{2} \leq R.y_{2}$)) \keyw{then} \\
3. \> \> \> \keyw{if} $(R'.z<u.Top_{v}[u.p].z)$ \keyw{then} $R'=u.Top_{v}[u.p]$ \\
4. \> \> \> \keyw{if} $(R'.z < u.H_{h}.z)$ \keyw{then} $R'=u.H_{h}$\\
5. \> \> \> \keyw{if} visible \keyw{then} RightReportRegions($R$,$visible$,$R'$,$u$) \\
6. \> \> \keyw{else} \\
7. \> \> \> \keyw{if} $(R'.z<u.Top_{v}[u.p].z)$ \keyw{then} $R'=u.Top_{v}[u.p]$ \\
8. \> \> \> \keyw{if} $(R'.z < u.H_{h}.z)$ \keyw{then} $R'=u.H_{h}$\\
9. \> \> \> \keyw{if} $(R.y_{1}<u.y_{mid})$ \keyw{then} RightEdge($R$,$visible$,$R'$,$lson(u)$) \\
10. \> \> \> \keyw{if} $(R.y_{2}>u.y_{mid})$ \keyw{then} RightEdge($R$,$visible$,$R'$,$rson(u)$) \\
11. \> \> $u.p=u.p+1$
\end{tabbing}
}
 \caption{This procedure is invoked when the right edge of a rectangle $R$ is encountered.}
 \label{fig:right}
\end{figure}

\begin{figure}
{\scriptsize
\begin{tabbing}
\quad \=\quad \=\quad \=\quad \=\quad \=\quad \=\quad \kill
\keyw{Procedure} RightReportRegions(Rectangle $R$, boolean $atR$, Rectagle $R'$, segment\_tree\_node $u$)  \\
1. \> \> \keyw{if} $(R.z<u.L)$ \keyw{then return} \\
2. \> \> \keyw{if} ($(u.H<R.z)$ AND $atR$) \keyw{then} \\
3. \> \> \> Find leaves $p$ and $q$ in $T$ so that $p.y<u.y_{1}$ and there is no other leaf $v$ such that $p.y<v.y<u.y_{1}$ \\ 
\> \> \> and act analogously for $u.y_{2}$ \\
4. \> \> \> Output region defined by $p.y_{1}$, $q.y_{2}$, $x$ coordinate of the sweep plane and $x$ coordinate of the field \\
\> \> \> showing the start of the region \\
5. \> \> \> Output regions defined by $p$ and $q$ and the leaves next to them in the list\\
6. \> \> \> \keyw{if} $(p.y_{1}=R.y_{1})$ \keyw{then} delete $p$ \\
7. \> \> \> \keyw{if} $(q.y_{2}=R.y_{2})$ \keyw{then} delete $q$ \\
8. \> \> \> Update properly the fields region of the adjacent leaves and the respective $x$ fields \\
9. \> \> \keyw{if} $(R'.z<u.Top_{v}[u.p].z)$ \keyw{then} $R'=u.Top_{v}[u.p]$ \\
10. \> \> \keyw{if} $(R'.z<u.H_{h}.z)$ \keyw{then} $R'=u.H_{h}$ \\
11. \> \> \keyw{if} $(u.H<R'.z)$ \keyw{then} \\
12. \> \> \> \keyw{if} $(u.y_{1}=R'.y_{1})$ \keyw{then} insert in $T$ the segment $s$ with $s.x_{1}=$ current $x$ and $s.y=R'.y_{1}$ \\
13. \> \> \> \keyw{if} $(u.y_{2}=R'.y_{2})$ \keyw{then} insert in $T$ the segment $s$ with $s.x_{1}=$ current $x$ and $s.y=R'.y_{2}$ \\
14. \> \> \> Find leaves $p$ and $q$ of $T$ such that $p.y \leq R'.y_{1} < R'.y_{2} \leq q.y$ and update field $p.region$ and the $x$ field \\
15. \> \> \keyw{else} \\
16. \> \> \> RightReportRegions($R$,$atR$,$R'$,$lson(u)$) \\
17. \> \> \> RightReportRegions($R$,$atR$,$R'$,$rson(u)$) 
\end{tabbing}
}
 \caption{This procedure is invoked by procedure $RightEdge$ to report the visible regions.}
 \label{fig:right1}
\end{figure}

 It is not hard to verify that the procedure $RightEdge$ updates appropriately the pointer $u.p$ while maintaining the visible segments (by using the flag $visible$). In Figure~\ref{fig:leftex}(b) the case handled by procedure $RightReportRegions$ is depicted. As before, we assume that the right edge of $R$ is divided into consecutive invisible and visible segments. Assume that the visible segments are $s(u_{1}),s(u_{2}),\ldots,s(u_{l})$ and the invisible segments are $s(w_{1}),s(w_{2}),\ldots,s(w_{m})$, where $u_{i}$ and $w_{i}$ are nodes of the segment tree. In addition, the visible pieces of rectangles along and below the right edge of $R$ are divided into basic segments $s(u'_{1}),s(u'_{2}),\ldots,s(u'_{r})$. The procedure $RightReportRegions$ is analogous to $LeftReportRegions$ except that it continues exploring below nodes $u_{i}$ to discover the new visible pieces. 
 
 This procedure maintains rectangle $R'$ to be the second highest rectangle after $R$ listed on the path from the root to the current node $u$. Each node $u'_{j}$ is a descendant of a visible node $u_{i}$, such that $u'_{j}.H$ is lower than the highest rectangle along the path from the root to $u'_{j}$. Therefore, we are in position to appropriately update $T$ with information we have obtained concerning the revealed rectangle $R'$ (note that $R'$ may as well be $background$). Thus, procedure $RightReportRegions$ explores a forest of trees of the segment tree. The roots of these subtrees are nodes that list $R$, each $u_{i}$ is contained in a subtree, the nodes $u'_{1},u'_{2},\ldots,u'_{r}$ are leaves of these subtrees and the remaining leaves are $w_{1},w_{2},\ldots,w_{m}$ (the proof can be found in \cite{bern}).
 
 The deletion of an edge affects the region tree in a similar way as the insertion of a vertical edge. First of all, the horizontal segments of the rectangle $R$ are removed from $T$ (if they were stored) and all the $region$ fields that belong to $R$ obtain the value $-\infty$ (reporting at the same time the respective regions). Then, the subtree rooted at a node $u_{i}$ with leaves $u'_{j}$ is traversed in an inorder fashion. Many basic segments $s(u'_{j})$ may belong to the same revealed rectangle $R'$. Instead of accessing the tree $T$ for each of the nodes $u'_{j}$ we save and combine the queries into one query. Because of the inorder tree walk, all the basic segments belonging to a single region of a revealed rectangle will be accessed sequentially. Thus, we have to access $T$ only once for each region of a revealed rectangle $R'$. This happens when we access a basic segment $s(u'_{j})$ which belongs to a region of a different revealed rectangle.

 \begin{figure}
 \centering
 \includegraphics[scale=0.6]{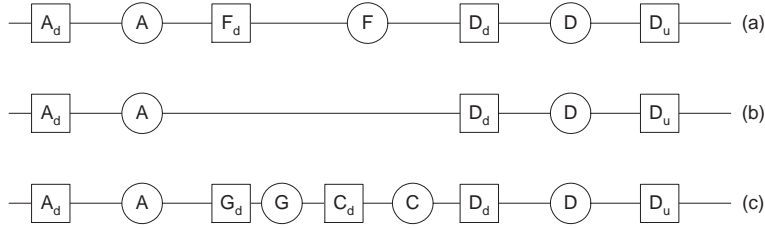}
 \caption{(a) a part of the list initially, (b) the same part after the deletion of segment $F_{d}$, (c) the list after the insertion of the segments of the revealed rectangles.}
 \label{fig:dlistright}
\end{figure}

 In Figure~\ref{fig:dlistright} an example of manipulation of $T$ is given when a right edge is encountered. Assume that the sweep plane is at position $x_{s}$ and has reached the right edge of the rectangle $F$. In a nutshell, the horizontal segment $F_{d}$ is removed from $T$ resulting in the list depicted in Figure \ref{fig:dlistright}(b). Note that $F_{u}$ is not visible and thus not present in the tree. When this segment is removed, the following regions are reported:
 
\begin{enumerate}

 \item{The region of $F$ defined by $F.y_{1}-D.y_{1}$ and $F.x-x_{s}$.}
 
 \item{The region of $A$ defined by $A.y_{1}-F.y_{1}$ and $A.x-x_{m}$.}

\end{enumerate}

 After the deletion, the $x$ coordinate of the start of the region $A$ obtains the value $x_{s}$. Assume that $C$ and $G$ are the revealed rectangles and that only their lower horizontal edges are visible. We insert first $G$ and then $C$ and finally the resulting list is depicted in Figure~\ref{fig:dlistright}(c).
 
\begin{lemma} \label{lem:right}
 Procedure $RightEdge$ requires $O(k\log{n})$ time to report $k$ visible regions for a scene of $n$ rectangles taking only into account the right edges.
\end{lemma}

\begin{proof}
 Assume that the visible segments of the right edge of a rectangle $R_{j}$ are $s(u_{1}),s(u_{2}),\ldots,s(u_{l_{j}})$ and the invisible segments $s(w_{1}),s(w_{2}),\ldots,s(w_{m})$, where $u_{i}$ and $w_{i}$ are nodes of the segment tree. The discovery of each $u_{i}$ node requires $O(\log{n})$ time. In addition, for each such node we update the region tree in $O(\log{n})$ time in order to report the respective visible region. In this way, the total time to report the visible regions for the right edge is $O(l_{j}\log{n})$. As a result, the time complexity for all rectangles will be $\sum_{j=1}^{n}{O(l_{j}\log{n})}$, which is bounded by $O(k\log{n})$.
 
 We must also consider the cost for computing and inserting in the region tree all revealed rectangles. Assume that a visible segment $s(u_{i})$ is divided into basic segments $s(u'_{1}),s(u'_{2}),\ldots,s(u'_{r})$. Each of these basic segments cost $O(\log{n})$ time to be inserted in $T$. However, this cost does not change the time complexity of the algorithm since each basic segment and the region it represents, will be reported later (in another sweep station) by either $LeftEdge$ or $RightEdge$. In the case of $LeftEdge$ the visible region which is represented by a basic segment will be reported by accessing $T$. Thus, the $O(\log{n})$ time overhead for each basic segment is assigned to the cost of reporting it. The same goes for the right edge.
\end{proof}

The following theorem summarizes the result.

\begin{theorem}
 The hidden surface removal problem for a set of $n$ iso-oriented rectangles can be solved in $O((n+k)\log{n})$ time and linear space, where $k$ is the number of reported regions.
\end{theorem}

\begin{proof}
 The time complexity of the algorithm is:\\
 
 {\em Total Time $=$ Prepr. $+$ (Precomp. of segm. tree)$\times$(\#slabs) $+$ (Reporting Time)}\\

 To sort the $x$, $y$ and $z$ coordinates, $O(n\log{n})$ time is required (Preprocessing). In each slab $O(n)$ time is necessary (Lemma~\ref{lem:con-seg}) to construct the segment tree and the arrays for each node. As a result, $O(n\log{n})$ time is needed in total because the scene is divided into $O(\log{n})$ slabs. From Lemma~\ref{lem:left} and Lemma~\ref{lem:right} we deduce that the reporting time is $O(k\log{n})$, where $k$ is the number of visible regions reported. From this discussion it is clear that the total time of the algorithm is $O((n+k)\log{n})$.
 
 The space complexity of the algorithm is: \\
 
 {\em Total Space $=$ (Space for segment tree) $+$ (Space for region tree)}\\

 The sequences $u.xTop_{v}$ , $u.Top_{v}$ and $u.xHL$ can be constructed (Lemma~\ref{lem:con-seg}) in linear time and so the space cannot be more. The skeleton of the segment tree requires linear space (since we store $O(n/\log{n})$ segments, each of which is associated with $O(\log{n})$ nodes). As a result, the total space needed by the segment tree is linear. For the region tree, the crucial observation is that at any position the sweep plane will intersect at most $2n$ horizontal segments. As a result, at most $2n-1$ regions can be visible in any sweep station of the sweep plane. Consequently, the region tree $T$ has at most $2n-1$ leaves and so it requires linear space. Therefore, the total space is $O(n)$.
\end{proof}

\section{Conclusions} \label{sec:con}

 In this paper we designed an algorithm for hidden surface removal of iso-oriented rectangles in a static scene. Our algorithm uses linear space and reports all visible regions in $O((n+k)\log n)$ time, where $n$ is the number of rectangles present in the scene and $k$ is the number of reported regions.
 
 The open problem is to design an $O(n\log n + k)$ algorithm that uses linear space for this problem. It would be also nice if these techniques could be transferred to more general scenes consisting of arbitrary rectangles or even polygons.

\end{document}